\documentclass[twocolumn,prl, superscriptaddress]{revtex4-1}
\usepackage[dvips]{graphicx}
\usepackage{amsmath,amssymb,amsthm,mathrsfs,amsfonts}
\usepackage{physics}
\usepackage{bm}
\usepackage{enumerate}
\usepackage{soul,xcolor}
\usepackage{verbatim}

\newtheorem{lemma}{Lemma}

\newcommand{\renyi}{R$\mathrm{\acute{e}}$nyi }
\newcommand{\mc}{\mathcal}
\begin{document}

\setstcolor{red}   

\title{One-Shot Coherence Dilution}
\author{Qi Zhao}
\author{Yunchao Liu}
\author{Xiao Yuan}
\email{yxbdwl@gmail.com}
\affiliation{Center for Quantum Information, Institute for Interdisciplinary Information Sciences, Tsinghua University, Beijing, 100084 China}
\author{Eric Chitambar}
\email{echitamb@siu.edu}
\affiliation{Department of Physics and Astronomy, Southern Illinois University, Carbondale, Illinois 62901, USA}
\author{Xiongfeng Ma}
\email{xma@tsinghua.edu.cn}
\affiliation{Center for Quantum Information, Institute for Interdisciplinary Information Sciences, Tsinghua University, Beijing, 100084 China}

\begin{abstract}
Manipulation and quantification of quantum resources are fundamental problems in quantum physics. In the asymptotic limit, coherence distillation and dilution have been proposed by manipulating infinite identical copies of states.  In the nonasymptotic setting, finite data-size effects emerge, and the practically relevant problem of coherence manipulation using finite resources has been left open.  This Letter establishes the one-shot theory of coherence dilution, which involves converting maximally coherent states into an arbitrary quantum state using maximally incoherent operations, dephasing-covariant incoherent operations, incoherent operations, or strictly incoherent operations.
We introduce several coherence monotones with concrete operational interpretations that estimate the one-shot coherence cost---the minimum amount of maximally coherent states needed for faithful coherence dilution.
Furthermore, we derive the asymptotic coherence dilution results with maximally incoherent operations, incoherent operations, and strictly incoherent operations as special cases.
Our result can be applied in the analyses of quantum information processing tasks that exploit coherence as resources, such as quantum key distribution and random number generation.
\end{abstract}

\maketitle

Quantum coherence is a fundamental property that can emerge within any quantum system.  With respect to some physically preferable reference frame \cite{Aharonov67,Kitaev04,Bartlett07}, such as the energy levels of an atom or a selected measurement basis, coherence empowers the ability of many quantum information tasks, including cryptography \cite{coles2016numerical}, metrology \cite{giovannetti2011advances}, and randomness generation \cite{Yuan15intrinsic, ma2017source}. Furthermore, coherence is a widespread resource playing important roles in biological systems \cite{plenio2008dephasing, rebentrost2009role} and small-scale thermodynamics \cite{Aberg14,Horodecki15,Lostaglio15,lostaglio2015description,narasimhachar2015low}.

Various efforts have been devoted to building a resource framework of coherence \cite{aberg2006quantifying, Baumgratz14, streltsov2016quantum}. In general, a resource theory is defined by a set of free states and a corresponding set of free operations that preserve the free states.  States that are not free are said to possess resource, and various measures can be constructed to quantify the amount of resource in a given state.  For example, in the resource theory of entanglement \cite{Bennett96, plenio2005introduction,Horodecki09}, free states and free operations are defined by separable states, local operation and classical communication (LOCC), respectively. Entanglement measures include the relative entropy of entanglement \cite{Vedral97} and entanglement of formation \cite{Bennett96}.

In the resource theory of coherence \cite{aberg2006quantifying, Baumgratz14}, free or incoherent states are those that are diagonal in \textit{a priori} fixed computational basis; free or incoherent operations are some specified classes of physically realizable operations that act invariantly on the set of incoherent states.  Different definitions of incoherent operations have been studied due to different motivations. In this work, we focus on the maximally incoherent operation (MIO) proposed by \AA berg \cite{aberg2006quantifying}, the dephasing-covariant incoherent operation (DIO) proposed independently by Chitambar and Gour \cite{Chitambar16prl} and Marvian and Spekkens \cite{Marvian16}, the incoherent operation (IO) proposed by Baumgratz \textit{et al.} \cite{Baumgratz14}, and the strictly incoherent operation (SIO) proposed by Winter and Yang \cite{Winter16}.
Coherence measures include the relative entropy of coherence \cite{Baumgratz14}, coherence of formation \cite{aberg2006quantifying, Yuan15intrinsic}, robustness of coherence \cite{Napoli16}, etc.
We refer to Ref.~\cite{streltsov2016quantum} for a comprehensive review of recent developments of the resource theory of coherence.

Investigating state transformations via free operations is of paramount importance in a resource theory. In particular, many efforts have been devoted to understand the interconversion between a given state $\rho$ and copies of a canonical unit resource $\ket{\Psi}$ \footnote{The unit resource $\ket{\Psi}$ is an EPR pair or a maximally coherent state in the resource theories of entanglement and coherence, respectively.} via free operations. Specifically, the dilution problem is to convert unit resource $\ket{\Psi}$ to the target state $\rho$, and the distillation problem is the reverse process. In the asymptotic case, where infinite copies of $\rho$ and $\ket{\Psi}$ are provided, the dilution rate (or coherence cost) and the distillable rate describe the maximal proportion of $\rho$ and  $\ket{\Psi}$ that can be obtained on average, respectively. In entanglement theory, the well-known distillable entanglement \cite{rains2001semidefinite} and entanglement cost \cite{hayden2001asymptotic} of a state measure its optimal rate of asymptotic distillation and dilution, respectively.  The asymptotic distillation and dilution of coherence under IO and SIO have been investigated by Winter and Yang \cite{Winter16} who proved that the distillable coherence is given by the relative entropy of coherence and that the coherence cost is given by the coherence of formation.

The processes of asymptotic distillation and dilution are studied under two crucial assumptions: (i) a source is available that prepares independent and identically distributed (IID) copies of the same state and (ii) an unbounded number of copies of this states can be generated.  These assumptions overlook possible correlations between different state preparations and they become unreasonable when only a finite supply of states are available.  In order to relax the two assumptions, it is necessary to consider the most general scenario, i.e., the one-shot scenario, where the conversion is from a general initial state to a general final state.  Such a scenario reflects realistic experimental setups where we only manipulate finite and correlated states. In many quantum information tasks, such as quantum key distribution \cite{tomamichel2012framework}, device independent processing \cite{Vazirani3432, Miller14}, thermodynamics \cite{aaberg2013truly, skrzypczyk2014work,horodecki2013fundamental,Aberg14,brandao2015second,Horodecki15,Lostaglio15,lostaglio2015description,narasimhachar2015low}, quantum channel capacity \cite{Renner06, mosonyi2009generalized, Buscemi10, Wang12}, and general resource theory \cite{gour2016quantum}, some analyses have already been conducted in the one-shot scenario. In particular, one-shot entanglement distillation and dilution have been investigated under LOCC \cite{Datta09, buscemi2010distilling, Buscemi11}, as well as nonentangling maps and operations that generate negligible amount of entanglement \cite{Brandao11}.  In thermodynamics, conversion under thermal operations is known only for qubit states \cite{Horodecki15}.  For coherence, the necessary and sufficient conditions for single-copy state transformations are known only for pure states \cite{Du15, Winter16, Zhu-16a} and single qubit mixed states \cite{Chitambar16prl,Chitambar16pra}.
Generally, one-shot coherence distillation and dilution of general quantum states have been left as open problems \cite{Winter16, streltsov2016quantum}.

In this Letter, we consider one-shot coherence dilution under four widely accepted incoherent operations: MIO, DIO, IO, and SIO. We first review the coherence framework by \AA berg \cite{aberg2006quantifying} and Baumgratz et al.~\cite{Baumgratz14}. Then, we introduce several coherence monotones for different incoherent operations. In addition, we define the one-shot coherence cost in the dilution process, and explicitly show that the optimal one-shot coherence cost is characterized by the introduced coherence monotones. Moreover, when applying our results to the asymptotic IID scenario, we obtain the coherence cost under MIO and show that it is equal to the relative entropy of coherence.
Similarly, we also derive the asymptotic coherence cost under IO and SIO and show that it equals to the coherence of formation, which is consistent with the results in \cite{Winter16}. Our main results are summarized in Table \ref{Table:results}.  We introduce and discuss the results in details below and provide the proofs in the appendices.
\begin{table}[thb]\centering
\footnotesize
\caption{Coherence dilution in the one-shot and asymptotic scenarios with MIO, DIO, IO, and SIO. The two columns ``one-shot'' and ``asymptotic'' denote the coherence measures in the one-shot and asymptotic scenarios, respectively.}   \label{Table:results}
	\begin{tabular}{cccc}
		\hline
		Operation & One-shot & Asymptotic\\
		\hline
		MIO &$C_{MIO}^{\varepsilon}\approx C_{\max}^{\varepsilon}$&$C_{MIO}^\infty=C_r$\\
		DIO &$C_{DIO}^{\varepsilon}\approx C_{\Delta,\max}^{\varepsilon}$&$C_{DIO}^\infty=C_r$\\
		IO & $C_{IO}^{\varepsilon}=C_0^{\varepsilon}$&$C_{IO}^\infty=C_f$\\
		SIO & $C_{SIO}^{\varepsilon}=C_0^{\varepsilon}$&$C_{SIO}^\infty=C_f$\\
				\hline
	\end{tabular}
\end{table}

\emph{Framework.}---Considering a computational basis $I = \{\ket{i}\}_{i=0}^{d-1}$ in a $d$-dimensional Hilbert space $\mathcal{H}_d$, incoherent states are defined as $\delta = \sum_{i=0}^{d-1}p_i\ket{i}\bra{i}$,
where $\{p_i\}$ is a probability distribution. We denote the set of incoherent states as $\mathcal{I}$.
The MIOs introduced in Ref.~\cite{aberg2006quantifying} are physical or completely positive trace preserving (CPTP) maps $\Lambda$ such that $\Lambda(\delta)\in\mathcal{I}$, $\forall\delta\in\mathcal{I}$.
A CPTP map $\Lambda$ is called a \emph{dephasing-covariant incoherent operation} (DIO) if $\Lambda[\Delta(\rho)]=\Delta[\Lambda(\rho)]$ for all $\rho$ \cite{Chitambar16prl, Marvian16}. Here, $\Delta(\rho) = \sum_i\ket{i}\bra{i}\bra{i}\rho\ket{i}$ is the dephasing channel, and clearly DIO is a subset of MIO.  Another subset of MIO are the \emph{incoherent operations} (IOs) \cite{Baumgratz14} which are CPTP maps that admit a Kraus operator representation $\Lambda(\rho) = \sum_n K_n\rho K_n^\dag$ with the $\{K_n\}$ being ``incoherent-preserving'' operators, that is ${K_n\delta K_n^\dag}/{p_n} \in \mathcal{I}$ for all $n$ and all $\delta\in\mathcal{I}$.  Here $p_n = \mathrm{Tr}\left[ K_n\rho K_n^\dag\right]$ is the probability of obtaining the $n$th outcome.
In general, when $\Lambda(\rho) = \sum_n K_n\rho K_n^\dag$ is an  incoherent operation, the Kraus operator can always be represented as $K_n=\sum_{i}c_i\ket{f(i)}\bra{i}$, where $f$ is a function on the index set and $c_i\in[0,1]$ \cite{Winter16}.
Finally, \emph{strictly incoherent operations} (SIOs) are CPTP maps admitting a Kraus operator representation $\Lambda(\rho) = \sum_n K_n\rho K_n^\dag$ such that both $\{K_n\}$ and $\{K_n^\dag\}$ are incoherent-preserving operators \cite{Winter16}. The relations among different incoherent operations are shown in Fig.~\ref{fig:IO}.

\begin{figure}[bht]
\centering \resizebox{6cm}{!}{\includegraphics[scale=1]{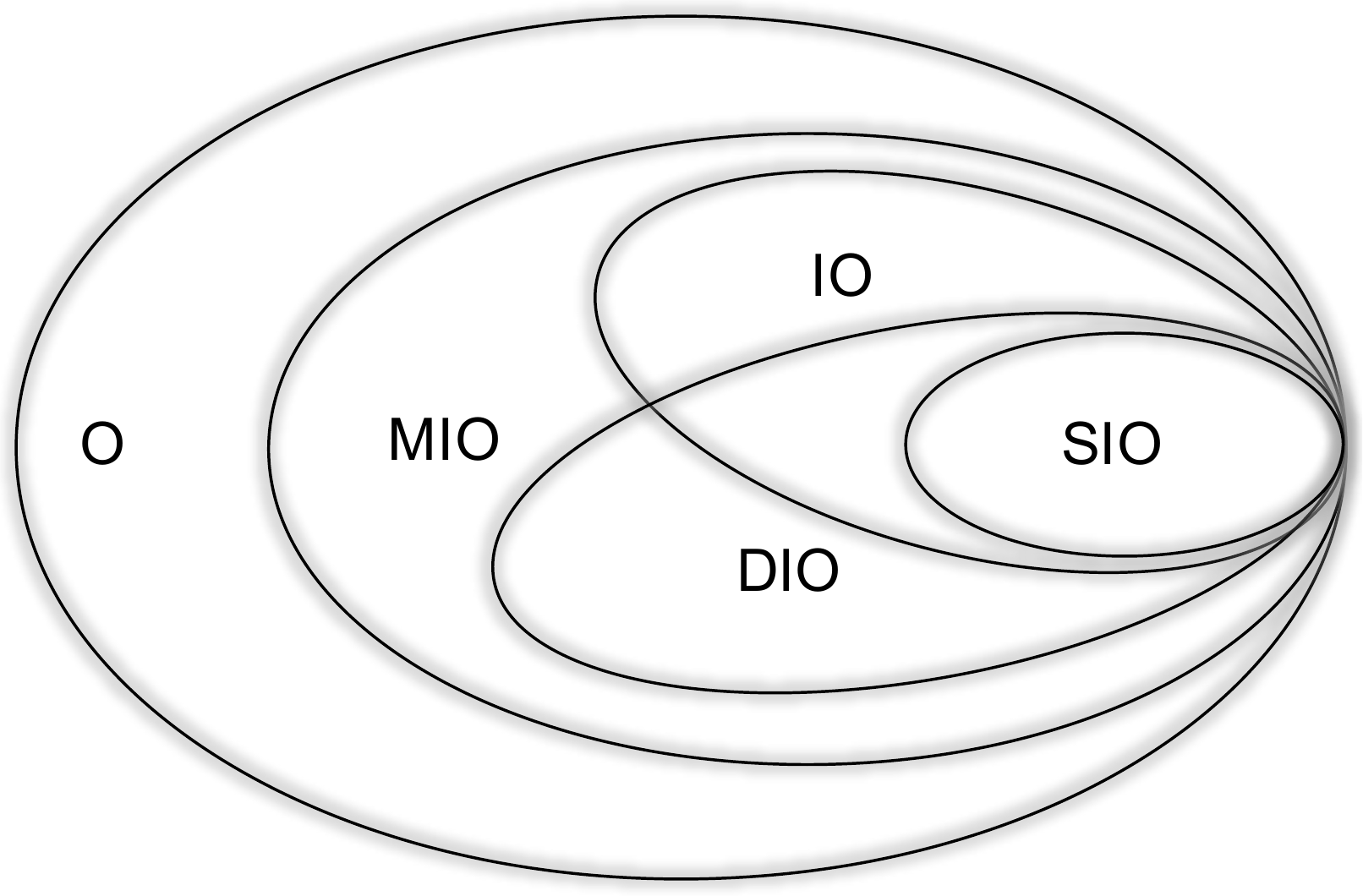}}
\caption{Comparison among different incoherent operations. The largest set $O$ contains all possible physical operations.} \label{fig:IO}
\end{figure}


Associated with each of these operational classes is a family of monotone functions.  A real-valued function $C(\rho)$ is called a MIO (DIO) monotone if $C(\rho)\geq C[\Lambda(\rho)]$ whenever $\Lambda$ is MIO (DIO). Since IO and SIO are defined in terms of Kraus operator representations, it is natural to modify the monotonicity condition to average postmeasurement values. That is, $C(\rho)$ is called an IO (SIO) monotone if $C(\rho)\geq \sum_n p_n C(K_n\rho_n K_n^\dagger/p_n)$ whenever $\{K_n\}$ ($\{K_n\}$ and also $\{K_n^\dagger\}$) are incoherent-preserving Kraus operators.

Following the notions of monotonicity defined above, a coherence measure for a class of operations $\mathcal{O}$ is defined as a real valued function $C(\rho)$ that satisfies the following requirements: (C1) $C(\rho)\ge 0$ with equality if and only if $\rho\in\mathcal{I}$; (C2) $C(\rho)$ is a monotone for operational class $\mathcal{O}$; (C3) Convexity: coherence cannot increase under mixing states,  i.e., $C\left(\sum_n p_n\rho_n\right) \le \sum_n p_nC(\rho_n)$. When a function satisfies conditions (C1) and (C2), we call it a coherence monotone for operational class $\mathcal{O}$. Although a coherence monotone may not satisfy (C3), it can still play important roles in tasks that process coherence.

\emph{Coherence monotones}---
In the following, we first introduce three coherence monotones of quantum states defined on Hilbert space $\mathcal{H}_d$. To do so, we make use of the { generalized quantum $\alpha$-\renyi divergence,
$\tilde{D}_{\alpha}(\rho||\sigma)=\frac{1}{\alpha-1}\log_2\left\{\Tr\left[\left(\sigma^{\frac{1-\alpha}{2\alpha}}\rho\sigma^{\frac{1-\alpha}{2\alpha}}\right)^{\alpha}\right]\right\}$,
where $\alpha\in(0,1)\cup(1,\infty)$ \cite{Lennert2013}}. For $\alpha=0,1,\infty$, the \renyi divergence is defined by taking the limit of $\alpha\to 0,1,\infty$, respectively.
Then, quantum coherence measures can be defined by
	$C_\alpha(\rho) = \min_{\delta\in S} \tilde{D}_{\alpha}(\rho||\delta)$, where $S\subset \mathcal{D}(\mathcal{H}_d)$ and $\mathcal{D}(\mathcal{H}_d)$ denotes the set of density matrices in $\mathcal{H}_d$ \cite{Rastegin16}. First we let $S = \mathcal{I}$. In the limit of $\alpha \rightarrow 1$, we recover the relative entropy of coherence
$C_{r}(\rho) = \min_{\delta\in\mathcal{I}}S(\rho||\delta)$. Here, $S(\rho||\delta) = \Tr[\rho\log_2\rho] - \Tr[\rho\log_2\delta]$ is the quantum relative entropy and  the minimization is over all incoherent states.
When $\alpha\rightarrow\infty$, we have the max-relative entropy $D_{\max}(\rho||\sigma)=\lim_{\alpha\to\infty}\tilde{D}_\alpha(\rho||\sigma)=\log_2\min\{\lambda|\rho\leq\lambda\sigma\}$. Then we introduce our first coherence monotone by
\begin{equation}\label{Eq:maxentropy}
  C_{\max}(\rho)=\min_{\delta\in\mathcal{I}}D_{\max}(\rho||\delta)
\end{equation}
With the properties of the max-relative entropy $D_{\max}(\rho||\sigma)$ \cite{Datta2009}, we can verify that $C_{\max}(\rho)$ satisfies (C1)---(C2). In addition, we show that it is quasi-convex, i.e.,
$  C_{\max}\left(\sum_ip_i\rho_i\right)\leq\max_iC_{\max}(\rho_i)$.

\emph{Theorem 1.}---The quantity $C_{\max}(\rho)$ is a coherence monotone under MIO \cite{Chitambar16pra} and it is quasiconvex.

We next consider a different set
\[A_\rho=\left\{\tfrac{1}{t}\left[(1+t)\Delta(\rho)-\rho\right]\;|\;t>0,\;(1+t)\Delta(\rho)-\rho\geq 0\right\}\]
 and let $\overline{A}_\rho$ denote its closure.  In particular, $\Delta(\rho)\in\overline{A}_\rho$ by taking the limit $t\to\infty$.  Analogous to Eq. \eqref{Eq:maxentropy}, we define $C_{\Delta, \max}=\min_{\sigma\in \overline{A}_\rho}D_{\max}(\rho||\sigma)$.
 The quantity $C_{\Delta, \max}$ was originally introduced in Ref. \cite{Chitambar16pra} and shown to have the simplified form
\begin{equation}
C_{\Delta, \max}(\rho)=\log_2\min\{\lambda\;|\;\rho\leq\lambda\Delta(\rho)\}.
\end{equation}

\emph{Theorem 2.}---$C_{\Delta, \max}$ is a coherence monotone under DIO \cite{Chitambar16pra} and it is quasiconvex.

Alternatively, another way of defining coherence measure is via the convex-roof construction.  For instance, the coherence of formation $C_f(\rho)$ can be defined by $C_{f}(\rho)\equiv\min_{\{p_j,\ket{\psi_j}\}}\sum_jp_jS(\Delta(\ket{\psi_j}\bra{\psi_j}))$ \cite{aberg2006quantifying,Yuan15intrinsic}.
Here, $S(\rho)=-\Tr[\rho\log_2\rho]$ is the Von-Neumann entropy, and the minimization is over all possible pure-state decompositions of $\rho = \sum_j p_j\ket{\psi_j}\bra{\psi_j}$.
Now, suppose $\ket{\psi_j}=\sum_{i=0}^{d-1}a_{ij}\ket{i}$ and denote $T_j$ to be the number of nonzero elements in $\{a_{0j},\cdots,a_{d-1,j}\}$.
We introduce our third coherence monotone $C_0(\rho)$,
\begin{equation}
	C_0(\rho)=\min_{\{p_j,\ket{\psi_j}\}}\max_{j}\log_2 T_j.
\end{equation}
Under this convex-roof construction, we show that  $C_0(\rho)$ is a coherence monotone. Nevertheless, it does not satisfy the convexity requirement (C3).

\emph{Theorem 3.}---$C_0(\rho)$ is a coherence monotone under IO and it violates the convexity requirement (C3).

\emph{One-shot dilution.}---
With the three coherence monotones, we are now ready to consider the process of coherence dilution which converts maximally coherent states into a target state.  The canonical maximally coherent state of dimension $M$ is given by $\ket{\Psi_M}=\frac{1}{\sqrt{M}}\sum_{i=0}^{M-1}\ket{i}$ \cite{Baumgratz14}.
One-shot coherence cost measures the minimal length $M$ such that $\ket{\Psi_M}$ can be converted into a target state $\rho$ via incoherent operations within some finite error. Based on different definitions of incoherent operations, we define different one-shot coherence costs.

\emph{Definition 1.}--- Let $\mathcal{O}\in\{MIO, DIO, IO, SIO\}$ denote a class of incoherent operations.  Then for a given state $\rho$ and $\varepsilon\ge 0$, the one-shot coherence cost under $\mathcal{O}$ is defined by
\begin{equation}\label{Eq:formationIO}
    C_{\mathcal{O}}^\varepsilon(\rho)=\min_{\Lambda\in \mathcal{O}}\{\log_2 M\vert F[\rho,\Lambda(\Psi_M)]\geq 1-\varepsilon\},
\end{equation}
where $F(\rho,\sigma) = \left(\Tr[\sqrt{\sqrt{\rho}\sigma\sqrt{\rho}}]\right)^2$ is the fidelity measure between two states $\rho$ and $\sigma$.

In the definition, there is a ``smoothing" parameter $\varepsilon$ in the dilution process, which is also known as the failure probability in cryptography \cite{renner2008security}. Instead of obtaining the exact state $\rho$, we allow the final state to be deviated no more than $\varepsilon$ from $\rho$, where the deviation is measured in fidelity. Specifically, when $\varepsilon = 0$, it becomes the case of perfect state conversion.

As shown in Fig.~\ref{fig:IO}, since $SIO\subset IO \subset MIO$ and $SIO\subset DIO \subset MIO$, we have $C_{MIO}^\varepsilon(\rho)\le  C_{IO}^\varepsilon(\rho) \le C_{SIO}^\varepsilon(\rho)$ and $C_{MIO}^\varepsilon(\rho) \le C_{DIO}^\varepsilon(\rho)\le C_{SIO}^\varepsilon(\rho)$ in general. However, IO and DIO are incomparable operations and we thus cannot derive a relationship between $C_{IO}^\varepsilon(\rho)$ and $C_{DIO}^\varepsilon(\rho)$ directly from the definitions.  Nevertheless, the hierarchy $C_{DIO}^\varepsilon(\rho)\leq C_{IO}^\varepsilon(\rho)$ can be established in the asymptotic case according to the following theorems.

To characterize coherence cost with certain error $\varepsilon$, we apply a smoothing to a general coherence measure $C(\rho)$ by minimizing over states $\rho'$ that satisfy $F(\rho,\rho')\geq 1- \varepsilon$,
\begin{equation}\label{}
    C^\varepsilon(\rho) = \min_{\rho': F(\rho,\rho')\geq 1-\varepsilon}C(\rho').
\end{equation}
We show that the one-shot coherence cost under MIO is bounded by the smoothed coherence measure $C_{\max}^{\varepsilon}(\rho)$,

\emph{Theorem 4.}---For any state $\rho$ and $\varepsilon\ge0$
\begin{equation}\label{Eq:maxthmMIO}
\begin{aligned}
  C_{\max}^{\varepsilon}(\rho)&\leq C_{MIO}^\varepsilon(\rho)\leq C_{\max}^{\varepsilon}(\rho)+1.\\
 \end{aligned}
\end{equation}
Similarly, the one-shot coherence cost using DIO is bounded by the smoothed coherence measure $C_{\Delta,\max}^{\varepsilon}(\rho)$,

\emph{Theorem 5.}---For any state $\rho$ and $\varepsilon\geq 0$,
\begin{equation}\label{Eq:maxDIO}
C^\varepsilon_{\Delta,\max}(\rho)\leq C^\varepsilon_{DIO}(\rho)\leq C^\varepsilon_{\Delta,\max}(\rho)+1.
\end{equation}
Finally, the one-shot coherence cost under IO and SIO is exactly characterized by the smoothed coherence measure $C_0^{\varepsilon}(\rho)$,

\emph{Theorem 6.}---For any state $\rho$ and $\varepsilon\ge0$
\begin{equation}\label{Eq:maxthmIO}
\begin{aligned}
C_{IO}^\varepsilon(\rho) = C_{SIO}^\varepsilon(\rho)=C_0^{\varepsilon}(\rho).
 \end{aligned}
  \end{equation}
The main proof idea of the theorems is to firstly prove the lower bound of the one-shot coherence cost by exploiting the monotonicity property of coherence measures. Then, the next step is to explicitly construct an incoherent operation such that this lower bound is saturated.
We leave the detailed proofs and the explicit transformations in the Appendix~\ref{theorem6}.

\emph{Asymptotic case.---}
Our one-shot coherence cost results hold for any state and any smooth parameter $\varepsilon$. As a special case, we can consider the asymptotic coherence dilution with an infinitely large number of i.i.d. target states.
We define the regularized coherence cost by taking the limit $n\rightarrow\infty$ and $\varepsilon\to 0^+$:
\begin{equation}\label{infformation}
  C_{\mathcal{O}}^\infty(\rho)=\lim_{\varepsilon\to 0^+}\lim_{n\to\infty}\frac{1}{n}C_{\mathcal{O}}^\varepsilon(\rho^{\otimes n}).
\end{equation}
where $\mathcal{O}\in\{MIO,IO,SIO,DIO\}$. Following the results of one-shot coherence dilution, we obtain coherence cost in the asymptotic case.

\emph{Theorem 7}---For any state $\rho$, the asymptotic coherence cost under MIO  is quantified by
\begin{equation}\label{Eq:CMIOasy}
\begin{aligned}
  C_{MIO}^\infty(\rho)=C_r(\rho).
  \end{aligned}
\end{equation}

Combining this with the work of Winter and Yang \cite{Winter16}, we see that both the asymptotic coherence cost under MIO and the the distillable coherence under IO is given by the relative entropy of coherence $C_r(\rho)$.  Since MIO is more powerful than IO, it follows that the distillable coherence under MIO is also characterized by $C_r(\rho)$.  One can also see this by noting that the converse proof for distillable coherence given in Ref. \cite{Winter16} also holds for MIO.  Thus, coherence is asymptotically reversible under MIO.  This is a slight strengthening of the general result presented in Ref. \cite{Brandao15} which implies reversibility by MIO when an asymptotically small amount of coherence can be generated.

Interestingly, we find that $C^\infty_{DIO}(\rho)=C_r(\rho)$, which is rather surprising since $C_{MIO}^0(\rho)\leq C_{DIO}^0(\rho)$, with the inequality being strict in many cases. Yet, evidently MIO and DIO yield the same coherence dilution rate in the asymptotic case, i.e.,
\begin{equation}\label{Eq:CIOasy}
\begin{aligned}
C_r(\rho) &=\lim_{\varepsilon\to 0^+}\lim_{n\to\infty}\frac{1}{n}C_{MIO}^\varepsilon(\rho^{\otimes n}) \\
&=\lim_{\varepsilon\to 0^+}\lim_{n\to\infty}\frac{1}{n}C_{DIO}^\varepsilon(\rho^{\otimes n}).
\end{aligned}
\end{equation}
The proof of this fact will be presented in a separate paper as it employs techniques quite different from the ones used in this work.

\emph{Theorem 8}---For any state $\rho$, the asymptotic coherence cost under IO and SIO is quantified by
\begin{equation}\label{Eq:CIOasy}
\begin{aligned}
   C_{IO}^\infty(\rho)=C_{SIO}^\infty(\rho)=C_f(\rho).
  \end{aligned}
\end{equation}
The asymptotic coherence dilution under IO and SIO has been investigated by Winter and Yang \cite{Winter16}. The coherence cost is characterized by the coherence of formation $C_f(\rho)$, which is consistent with our result. Note that, although the problems of one-shot and asymptotic coherence dilutions are similar, the methods are different. Our method holds for any state and any $\varepsilon$ while the method by Winter and Yang holds only for infinite copies of the same state and the limit with $\varepsilon\to 0^+$. Furthermore, the definition in Eq.~\eqref{infformation} can also be generalized to
\begin{equation}
  C_{\mathcal{O}}^{\infty,\varepsilon}(\rho)=\lim_{n\to\infty}\frac{1}{n}C_{\mathcal{O}}^\varepsilon(\rho^{\otimes n})~~ \varepsilon\in (0,1).
\end{equation}
By applying the property of the quantum asymptotic equipartition \cite{tomamichel2012framework,tomamichel2009fully,2016arXiv160701796D}, we may also obtain the same results in Theorem 7 and 8.
It would be an interesting future work to propose the generalized asymptotic coherence dilution with a finite smooth parameter $\varepsilon$ and relate it to the corresponding coherence monotone.

\emph{Discussion.}---
Our work solves the open problem of one-shot coherence dilution \cite{streltsov2016quantum} and derive the conventional coherence dilution formula in the asymptotic limit. Our results also indicate that coherence is asymptotically reversible under MIO and DIO. According to recent investigations of the resource theory of coherence \cite{streltsov2016quantum}, our results also shed light on the role of coherence as a resource in quantum information processing tasks like random number generation \cite{Yuan15intrinsic, ma2017source} and cryptography \cite{coles2016numerical}.

In the asymptotic scenario, the distillable coherence and coherence cost are additive, i.e., $C_f(\rho_1\otimes\rho_2) = C_f(\rho_1)+C_f(\rho_2)$ and $C_r(\rho_1\otimes\rho_2) = C_r(\rho_1)+C_r(\rho_2)$ \cite{Winter16}. In contrast, the proposed one-shot coherence monotones do not satisfy the additivity property in general. For example, one can show that $C_0^\varepsilon(\rho^{\otimes n}) \neq C_0^\varepsilon(\rho)+C_0^\varepsilon(\rho^{\otimes n-1})$ when $n\to\infty$ and $\varepsilon\to 0^+$. Furthermore, given $\varepsilon\geq \varepsilon_1+\varepsilon_2$, we can derive that $C_0^{\varepsilon}(\rho_1\otimes\rho_2)\leq C_0^{\varepsilon_1}(\rho_1)+C_0^{\varepsilon_2}(\rho_2)$. The inequality also holds for $C_{\max}^{\varepsilon}$ and $C_{\Delta,\max}^{\varepsilon}$ with proofs shown in Appendix~\ref{additi}. An interesting future direction is to study general and tight additive inequalities for these coherence monotones.

Another interesting perspective is to quantify the one-shot coherence distillation under different incoherent operations. The essential problem is to find the conversion from a general mixed state to the maximally coherent state. Some interesting results have been obtained \cite{regula2017one} but the general results still remain to be solved.
Furthermore, due to the strong similarity, we also expect that our result can shed light on the one-shot coherence conversion under thermal operations in the thermodynamic scenario, which has been partially solved only for qubits recently \cite{Horodecki15}.

We acknowledge F.~Buscemi, T.~Peng, and A.~Winter for the insightful discussions. This work was supported by the National Key R$\&$D Program of China Grant No. 2017YFA0303900 and 2017YFA0304004, the National Natural Science Foundation of China Grant No. 11674193 and the National Science Foundation (NSF) Early CAREER Award No.1352326. Q.~Z., Y.~L.~and X.~Y.~contributed equally to this Letter.


\bibliographystyle{apsrev4-1}
\bibliography{BibCoherence3}

\onecolumngrid
\appendix

\section{Proof of Theorem 1, 2}
\begin{proof}
We refer to \cite{Chitambar16pra,Datta2009} for the proof that $C_{\max}(\rho)$ and $C_{\Delta,\max}(\rho)$ are coherence monotones under MIO and DIO satisfying (C1) and (C2), respectively. Also, we know from \cite{Datta2009} that $D_{\max}(\rho||\sigma)$ is quasi-convex, i.e.
\begin{equation}\label{quasiconvex}
  \begin{aligned}
D_{\max}\left(\sum_ip_i\rho_i\Big|\Big|\sum_ip_i\sigma_i\right)\leq\max_iD_{\max}(\rho_i||\sigma_i).
  \end{aligned}
\end{equation}
Then for $\rho=\sum_i p_i \rho_i$, suppose $\min_{\delta\in\mathcal{I}}D_{\max}(\rho_i||\delta)=D_{\max}(\rho_i||\delta_i^\ast)$, we have
\begin{equation}\label{quasiconvex2}
  \begin{aligned}
C_{\max}\left(\sum_i p_i\rho_i\right)&=\min_{\delta\in\mathcal{I}} D_{\max}\left(\sum_i p_i\rho_i\Big|\Big|\delta\right)\\
&\leq D_{\max}\left(\sum_i p_i\rho_i\Big|\Big|\sum_i p_i \delta_i^\ast\right)\\
&\leq \max_iD_{\max}\left(\rho_i||\delta_i^\ast\right)\\
&= \max_iC_{\max}(\rho_i),
  \end{aligned}
\end{equation}
and we conclude that $C_{\max}(\rho)$ is quasi-convex. The same argument also holds for $C_{\Delta,\max}(\rho)$.
\end{proof}

\section{Proof of Theorem 3}
In order to accomplish our proof, we first review a Lemma proposed in \cite{Winter16}.
\begin{lemma}\label{pure}
If $\ket{\psi}=\frac{ K_i\ket{\phi}}{\sqrt{\Tr[K_i\ket{\phi}\bra{\phi}K_i^\dag]}}$ where $\{K_i\}$ is a set of incoherent-preserving Kraus operators, then $C_0(\ket{\psi}\bra{\psi})\leq C_0(\ket{\phi}\bra{\phi})$.
\end{lemma}

We will prove that $C_{0}(\rho)$ is a coherence monotone under both MIO and IO, but not convex.

\begin{proof}

(C1) Suppose $C_0(\rho)=0$ and the corresponding decomposition of $\rho$ is $\{p_j,\ket{\psi_j}\}$. Since $T_j=1$ for all $j$, we know that $\ket{\psi_j}=\ket{j'}$, which means that $\rho\in\mathcal{I}$. Conversely, suppose $\rho=\sum_i\delta_i\ket{i}\bra{i}$, we can choose $\{\delta_i,\ket{i}\}$ as a decomposition which leads to $C_0(\rho)=0$.

(C2, MIO) Let $\{p_j,\ket{\psi_j}\}$ be the decomposition such that $C_0(\rho)=\max_j\log_2 T_j$. Let $M=\max_jT_j$. Let $\Lambda$ be any maximally incoherent operation with $\Lambda(\rho)=\sum_{n}K_n\rho K_n^\dagger$, and $\sum_{n} K_n^\dagger K_n=\mathcal{I}$. Then
\begin{equation}\label{C2}
  \Lambda(\rho)=\sum_{i}\sum_{n}p_i K_n\ket{\psi_i}\bra{\psi_i}K_n^\dag.
\end{equation}
We choose an ensemble of $\Lambda(\rho)$ as
\begin{equation}\label{ensemble}
  \left\{p_i\Tr[K_n\ket{\psi_i}\bra{\psi_i}K_n^\dag],\frac{ K_n\ket{\psi_i}}{\sqrt{\Tr[K_n\ket{\psi_i}\bra{\psi_i}K_n^\dag]}}\right\}
\end{equation}
From Lemma \ref{pure} we know that $\max_{i,n}\log_2 T_{in}\leq\log_2 M$. Since we take the minimum over all decompositions, we conclude that $C_0(\Lambda(\rho))\leq C_0(\rho)$.

(C2, IO) Let $\{p_j,\ket{\psi_j}\}$ be the decomposition such that $C_0(\rho)=\max_j\log_2 T_j=\log_2 M$. Let $\Lambda$ be any incoherent operation with $\Lambda(\rho)=\sum_{n}K_n\rho K_n^\dagger$. Then the post-measurement state of $n$th outcome is
\begin{equation}\label{C3}
  \rho_n=\frac{1}{\Tr[K_n\rho K_n^\dag]}\sum_{i}p_i K_n\ket{\psi_i}\bra{\psi_i}K_n^\dag.
\end{equation}
We choose an ensemble of $\rho_n$ as
\begin{equation}\label{ensemble2}
  \left\{p_i\frac{\Tr[K_n\ket{\psi_i}\bra{\psi_i}K_n^\dag]}{\Tr[K_n\rho K_n^\dag]},\frac{ K_n\ket{\psi_i}}{\sqrt{\Tr[K_n\ket{\psi_i}\bra{\psi_i}K_n^\dag]}}\right\}
\end{equation}
From Lemma \ref{pure} we know that $\max_{i}\log_2 T_i\leq\log_2 M$. Since we take the minimum over all decompositions, we conclude that for any $n$, $C_0(\rho_n)\leq C_0(\rho)$. Therefore
\[\sum_{n}p_nC_0(\rho_n)\leq C_0(\rho).\]

(C3) We give a counter example to show that $C_0$ does not satisfy the convex property.
We take $p_1=p_2=\frac{1}{2}$, $\rho_1=\frac{1}{3}(\ket{0}+\ket{1}+\ket{2})(\bra{0}+\bra{1}+\bra{2})$, $\rho_2=\frac{1}{3}(\ket{0}\bra{0}+\ket{1}\bra{1}+\ket{2}\bra{2})$, and $\rho=p_1\rho_1+p_2 \rho_2$.
Consequently, we have $C_0(\rho_1)=\log_2 3$ and $C_0(\rho_2)=\log_2 1=0 $. Moreover, $C_0(\rho)\ge \log_22$ because for each decomposition $\{p_j,\ket{\psi_j}\}$ of $\rho$ we have $\max_{j} T_j\ge 2$, since otherwise $\rho$ becomes a mixture of incoherent states, contradicting the fact that $\rho$ has non-diagonal terms. Thus
\begin{equation}
p_1C_0(\rho_1)+ p_2C_0(\rho_2)=\frac{1}{2}\log_23<\log_22\leq C_0(\rho),
\end{equation}
which is contradict to the convexity requirement.
\end{proof}

\section{Proof of Theorem 4}
\begin{proof}
First we prove the left hand side, let $\log_2 M=C_{MIO}^\varepsilon(\rho)$ and $\sigma=\sum_{i=0}^{M-1}\frac{1}{M}\ket{i}\bra{i}$. The definition of $C_{MIO}^\varepsilon(\rho)$ implies that there exists an operation $\Lambda\in MIO$ such that
  $F(\rho,\Lambda(\Psi_M))\geq 1-\varepsilon$ and we let $\rho'=\Lambda(\Psi_M)$. Then we have
\begin{equation}\label{maxleftbound}
  \begin{aligned}
  C_{\max}^\varepsilon(\rho)&\leq C_{\max}(\rho')\\
  &=\min_{\delta\in\mathcal{I}}D_{\max}(\Lambda(\Psi_M)||\delta)\\
  &\leq D_{\max}(\Lambda(\Psi_M)||\Lambda(\sigma))\\
  &\leq D_{\max}(\Psi_M||\sigma)\\
  &=\log_2 M.
  \end{aligned}
\end{equation}
For the right hand side, we choose the state $\rho'$ reaching minimum such that $C_{\max}^{\varepsilon}(\rho)=C_{\max}(\rho')=D_{\max}(\rho'||\delta)=\log_2 \lambda$ with $F(\rho,\rho')\geq 1-\varepsilon$. Let $M=\lceil\lambda\rceil$, then $M\delta\geq
\rho'$. Consider the following map
\begin{equation}\label{maxmap}
\begin{aligned}
  \Lambda(\omega)=&\frac{M}{M-1}\left(\Tr[\Psi_M\omega]-\frac{1}{M}\right)\rho'\\
  &+\frac{M}{M-1}(1-\Tr[\Psi_M\omega])\delta.
  \end{aligned}
\end{equation}
Since $\Tr[\Psi_M\delta]=\frac{1}{M}$ for all $\delta\in\mathcal{I}$, we know that $\Lambda(\delta)\in\mathcal{I}$ for all $\delta\in\mathcal{I}$, thus $\Lambda\in MIO$. We can rewrite the map as
\begin{equation}\label{maxmaprewrite}
\begin{aligned}
  \Lambda(\omega)=&\frac{M}{M-1}(1-\Tr[\Psi_M\omega])\left(\delta-\frac{1}{M}\rho'\right)+\Tr[\Psi_M\omega]\rho'.
  \end{aligned}
\end{equation}
Since $\delta\geq \frac{1}{M}\rho'$, we know that $\Lambda$ is completely positive. Notice that $\Lambda(\Psi_M)=\rho'$, we have
\begin{equation}\label{maxrightbound}
  \begin{aligned}
  C_{MIO}^\varepsilon(\rho)&\leq \log_2 M\\
  &\leq \log_2(1+\lambda)\\
  &\leq 1+\log_2\lambda\\
  &=C_{\max}^{\varepsilon}(\rho)+1.
  \end{aligned}
\end{equation}
\end{proof}

\section{Proof of Theorem 5}
\begin{proof}
  The proof follows just like the MIO case. Let $\log_2 M=C_{DIO}^\varepsilon(\rho)$. The definition of $C_{DIO}^\varepsilon(\rho)$ implies that there exists an operation $\Lambda\in DIO$ such that
  $F(\rho,\Lambda(\Psi_M))\geq 1-\varepsilon$ and we let $\rho'=\Lambda(\Psi_M)$.  For the lower bound on $C^\varepsilon_{DIO}(\rho)$, we have
\begin{align}
C^\varepsilon_{\Delta,\max}(\rho)&\leq C_{\Delta,\max}(\rho')\notag\\
&=\min_{\delta\in \overline{A}_{\rho'}}D_{\max}(\Lambda(\Psi_M)||\delta)\notag\\
&\leq D_{\max}(\Lambda(\Psi_M)||\Delta(\Lambda(\Psi_M)))\notag\\
&= D_{\max}(\Lambda(\Psi_M)||\Lambda(\Delta(\Psi_M)))\notag\\
&\leq D_{\max}(\Psi_M||\Delta(\Psi_M))\notag\\
&=\log_2 M,
\end{align}
where we use the facts that (i) $\Delta(\rho')=\Delta(\Lambda(\Psi_m))\in \overline{A}_{\rho'}$, and (ii) $\Lambda$ commutes with $\Delta$.

For the upper bound on $C^\varepsilon_{DIO}(\rho)$, let $\rho'$ be the minimizing state such that $C^\varepsilon_{\Delta,\max}(\rho)=C_{\Delta,\max}(\rho')=\log_2\lambda$, where $\rho'\leq\lambda\Delta(\rho')$.  Let $M=\lceil\lambda\rceil$ and define the map
\begin{align}
\mc{E}(\omega)&=\frac{M}{M-1}\left(\left[\Tr(\Psi_M\omega)-\frac{1}{M}\right]\rho'+\left[1-\Tr(\Psi_M\omega)\right]\Delta(\rho')\right)\notag\\
&=\frac{M}{M-1}\left(\left[1-\Tr(\Psi_M\omega)\right]\left(\Delta(\rho')-\frac{1}{M}\rho'\right)\right)+\Tr(\Psi_M\omega)\rho'.
\end{align}
This map is CP since $\rho'\leq M\Delta(\rho')$, and we can see that it is dephasing-covariant since
\begin{align}
\mc{E}(\Delta(\omega))&=\frac{M}{M-1}\left(\left[1-\Tr(\Psi_M\Delta(\omega))\right]\left(\Delta(\rho')-\frac{1}{M}\rho'\right)\right)+\Tr(\Psi_M\Delta(\omega))\rho'\notag\\
&=\frac{M}{M-1}\left(\left[1-\frac{1}{M}\right]\left(\Delta(\rho')-\frac{1}{M}\rho'\right)\right)+\frac{1}{M}\rho'=\Delta(\rho')\\
\Delta(\mc{E}(\omega))&=\frac{M}{M-1}\left(\left[1-\Tr(\Psi_M\omega)\right]\left(\Delta(\rho')-\frac{1}{M}\Delta(\rho')\right)\right)+\Tr(\Psi_M\omega)\Delta(\rho')=\Delta(\rho').
\end{align}
Finally note that $\mc{E}(\Psi_M)=\rho'$ by construction, and therefore
\[C_{DIO}^\varepsilon(\rho)=\log_2 M\leq 1+ \log_2\lambda=C_{\Delta,\max}(\rho')+1=C^\varepsilon_{\Delta,\max}(\rho)+1.\]
\end{proof}

\section{Proof of Theorem 6}\label{theorem6}

\begin{proof}
First we study $IO$. For the lower bound on $C_{IO}^\varepsilon(\rho)$, let $\log_2 M=C_{IO}^\varepsilon(\rho)$, then there exists an operation $\Lambda \in IO $  such that
  $F(\rho,\Lambda(\Psi_M))\geq 1-\varepsilon$ and we let $\rho'=\Lambda(\Psi_M)$. Then we have

\begin{equation}\label{maxIOleftbound}
  \begin{aligned}
  C_{0}^\varepsilon(\rho)&\leq C_{0}(\rho')\\
  &= C_{0}(\Lambda(\Psi_M)) \\
  &\leq C_{0}(\Psi_M) \\
  &=\log_2 M= C_{IO}^\varepsilon(\rho).
  \end{aligned}
\end{equation}

For the other direction, assume that $\rho'$ is the state reaching minimum such that $C_{0}^\varepsilon(\rho)=C_{0}(\rho')$. Let $\log_2 M'= C_{0}(\rho')$. We will show that there exist a $\Lambda\in IO$ such that $F(\rho,\Lambda(\Psi_M'))\geq 1-\varepsilon$.

Let $C_{0}(\rho')=\max_j \log_2 T_j$ with the corresponding ensemble $\{p_j,\ket{\psi_j}\}$. Then $M'=\max_jT_j$. Without loss of generality, for a given $j$, let $\ket{\psi_j}=\sum_i a_{ji}\ket{i}$, we can assume that $|a_{ji}|^2 $ is in nonincreasing order, i.e. $|a_{j1}|^2\ge |a_{j2}|^2\ge \cdots \ge |a_{jd}|^2 $ and denote that $ \bm{\psi}_j=(|a_{j1}|^2,\cdots,|a_{jd}|^2)^T$ and $\bm{m}=(1/M',\cdots,1/M')^T$. Notice that the sequence $\frac{k}{\sum_{i=1}^{k}|a_{ji}|^2 }$ $(k=1,\cdots,T_j)$ is an increasing sequence due to the nonincreasing order of $|a_{ji}|^2 $, that is,
\begin{equation}
\begin{aligned}
\frac{1}{|a_{j0}|^2}\le \cdots\le \frac{T_j}{\sum_{i=1}^{T_j}|a_{ji}|^2 }=T_j\le M'.
\end{aligned}
\end{equation}
Consequently, we have for any $k\in \{1,\cdots,T_j\}$, $\frac{k}{M'}\le \sum_{i=1}^{k}|a_{ji}|^2$, which implies the majorization relation \[\bm{m}\prec \bm{\psi}_j.\] Thus there is a probability distribution $\{\lambda_{\pi}^j\}$ over permutations $\pi$ such that
\begin{equation}
\begin{aligned}
\bm{m}=\sum_{\pi} \lambda_{\pi}^j \bm{\psi}_j^\pi
\end{aligned}
\end{equation}
where $\bm{\psi}_j^\pi$ is the vector $\bm{\psi}_j$ with indices permuted according to $\pi: \bm{\psi}_j^\pi(i)=  |a_{j\pi(i)}|^2$\cite{Winter16,Du15}.
We construct the following Kraus operator
\begin{equation}
\begin{aligned}
K_{\pi}^j=\sum_i \sqrt{p_j\lambda_{\pi}^j} \sqrt{\frac{|a_{j\pi(i)}|^2}{1/M'}} \ket{\pi(i)}\bra{i},
\end{aligned}
\end{equation}
which is incoherent-preserving and satisfying
\begin{equation}
\begin{aligned}
\sum_{\pi,j} {K_{\pi}^j}^\dag K_{\pi}^j &=\sum_{i,\pi,j} p_j\lambda_{\pi}^j \frac{|a_{j\pi(i)}|^2}{1/M'} \ket{i}\bra{i}\\
&=\sum_{i,j} p_j \ket{i}\bra{i} \\
&=\sum_{i} \ket{i}\bra{i}=\mathbb{I}.
\end{aligned}
\end{equation}

Let $\Lambda$ be the incoherent operation with Kraus operators $\{K_{\pi}^j\}$, then we have $\Lambda(\Psi_M')=\sum_j p_j \ket{\psi_j}\bra{\psi_j} =\rho'$ and  $F(\Lambda(\Psi_M'),\rho)=F(\rho',\rho)\ge 1-\varepsilon$. Thus $C_{IO}^\varepsilon (\rho)\le C_0^\varepsilon (\rho)$, combined with the previous fact $C_0^\varepsilon (\rho) \le C_{IO}^\varepsilon (\rho)$, we conclude that $C_0^\varepsilon (\rho) = C_{IO}^\varepsilon (\rho)$.

For the $SIO$ case, first observe that $C_{IO}^\varepsilon(\rho)\leq C_{SIO}^\varepsilon(\rho)$, which gives us the lower bound $C_{0}^\varepsilon(\rho)\leq C_{SIO}^\varepsilon(\rho)$. Also notice that the Kraus operators constructed above satisfy the definition for $SIO$, therefore the other direction $C_{0}^\varepsilon(\rho)\geq C_{SIO}^\varepsilon(\rho)$ also holds.
\end{proof}

\section{Proof of Theorem 7}
We prove the Theorem by introducing Lemma \ref{5p}, which is Proposition II.1 in \cite{Brandao2010}. First we consider a family of sets $\{M_n\}_{n\in\mathbb{N}}$ with $M_n\subseteq \mathcal{D}(\mathcal{H}^{\otimes n})$ that satisfies the following Properties.
\begin{enumerate}
  \item Each $M_n$ is convex and closed.
  \item Each $M_n$ contains $\sigma^{\otimes n}$, for a full rank state $\sigma\in \mathcal{D}(\mathcal{H})$.
  \item If $\rho\in M_{n+1}$ then $\Tr_k[\rho]\in M_n$, for every $k\in\{1,\cdots,n+1\}$.
  \item If $\rho\in M_n$ and $\nu\in M_m$, then $\rho\otimes\nu\in M_{n+m}$.
  \item If $\rho\in M_n$, then $P_\pi\rho P_\pi\in M_n$ for every permutation $\pi$ of length $n$.
\end{enumerate}
Here $P_\pi$ denotes the permutation in $\mathcal{H}^{\otimes n}$.
\begin{lemma}\label{5p}
    For every family of sets $\{M_n\}_{n\in\mathbb{N}}$ satisfying Properties 1-5 and every state $\rho\in\mathcal{D}(\mathcal{H})$,
    \begin{equation}\label{steinslemma}
      \lim_{n\to\infty}\frac{1}{n}\min_{\sigma\in M_n}S(\rho^{\otimes n}||\sigma)=\lim_{\varepsilon\to 0^+}\lim_{n\to\infty}\frac{1}{n}\min_{\rho'\in B^{\varepsilon}(\rho^{\otimes n})}\min_{\sigma\in M_n}D_{\max}(\rho'||\sigma),
    \end{equation}
    where $B^{\varepsilon}(\rho)=\{\rho'|F(\rho,\rho')\geq 1-\varepsilon\}$.
\end{lemma}
Now we prove Theorem 7 by applying Lemma \ref{5p}.
\begin{proof}
Let $\mathcal{I}$ denote the set of incoherent states in $\mathcal{H}$. Let $M_n=\mathcal{I}^{\otimes n}$ which indeed satisfies Properties 1-5. By rewriting Eq. \eqref{steinslemma} we obtain
  \begin{equation}\label{stein}
    \lim_{\varepsilon\to 0^+}\lim_{n\to\infty}\frac{1}{n}C_{\max}^\varepsilon(\rho^{\otimes n})=\lim_{n\to\infty}\frac{1}{n}C_r(\rho^{\otimes n}).
  \end{equation}
  From \cite{Winter16} (Theorem 9) we know that $C_r(\rho)$ is additive, so $C_r(\rho^{\otimes n})=nC_r(\rho)$. Finally, using Theorem 4 we conclude that
    \begin{equation}\label{stein}
    C_{MIO}^{\infty}(\rho)=\lim_{\varepsilon\to 0^+}\lim_{n\to\infty}\frac{1}{n}C_{\max}^\varepsilon(\rho^{\otimes n})=C_r(\rho).
  \end{equation}
\end{proof}

%
%

\section{Proof of Theorem 8}
By Theorem 6, we only need to prove the $IO$ case. We combine Lemma \ref{geq} and Lemma \ref{leq} to reach the result $C_{IO}^\infty(\rho)=C_f(\rho)$. The proof of Lemma \ref{geq} is similar to Lemma 5 in \cite{Buscemi11}.
\begin{lemma}\label{geq}
\begin{equation}\label{lemmainf}
  C_{IO}^\infty(\rho)\geq C_f(\rho).
\end{equation}
\end{lemma}
\begin{proof}
  $\forall \varepsilon>0$ and $n\in\mathbb{N}$ we have
  \begin{equation}\label{prooflemma}
  \begin{aligned}
    \frac{1}{n}C_{IO}^{\varepsilon}(\rho^{\otimes n})&=\frac{1}{n}C_0^\varepsilon(\rho^{\otimes n})\\
    &=\frac{1}{n}C_0(\rho')\\
    &\geq \frac{1}{n}C_f(\rho')\\
    &\geq \frac{1}{n}\min_{M_E}H(A|E)-O(\varepsilon)-O(1/n),
  \end{aligned}
  \end{equation}
  where the first line follows from Theorem 6, the second line follows by choosing the optimal state $\rho'$ in the neighborhood of $\rho^{\otimes n}$, the third line follows from the fact that $C_0(\rho)\geq C_f(\rho)$ which is shown in Ref.~\cite{liu2018quantum} and the last line follows from Fannes-Audenaert inequality \cite{fannesinequality}.

  Notice that $\min_{M_E}H(A|E)$ corresponds to the coherence of formation \cite{interplay},
  \begin{equation}\label{coherenceformation}
    \min_{M_E}H(A|E)=\min_{\{p_i,\ket{\psi_i}\}}\sum_i p_i S(\Delta(\ket{\psi_i}\bra{\psi_i}))=C_f(\rho^{\otimes n}),
  \end{equation}
  and that $C_f(\rho)$ is additive \cite{Winter16}, taking the limit $n\to+\infty$ and $\varepsilon\to 0^+$ we obtain Eq.~\eqref{lemmainf}.
\end{proof}

\begin{lemma}\label{leq}
  \begin{equation}\label{lemmainf2}
  C_{IO}^\infty(\rho)\leq C_f(\rho).
\end{equation}
\end{lemma}

\begin{proof}
The proof here is similar to the one given by \cite{Winter16} although from a different view of $C_0^\varepsilon$. The key idea is to construct a state $\rho''$ that is $\varepsilon$ close to $\rho^{\otimes n}$ and calculate the one-shot coherence measure $C_0$ of $\rho''$.

Given a decomposition of $\rho = \sum_j p_j\ket{\psi_j}\bra{\psi_j}$,  we construct a state that is $\epsilon$-close to $\rho^{\otimes n}=\sum_{\vec{j}=j_1j_2\cdots j_n}p_{\vec{j}}\ket{\psi_{\vec{j}}}\bra{\psi_{\vec{j}}}$. First, we consider the set of typical sequence,
\begin{equation}
	\mathcal{T}^1=\{\vec{j}:|N_j(\vec{j})/n-p_j|\le \delta_1, \forall j \},
\end{equation}
where $N_j(\vec{j})$ equals the number of $j_t$ in $\vec{j}$ such that $j_t=j, \forall t\in\{1,2,\cdots,n\}$. Then, we can define
\begin{equation}
	\rho' = \sum_{\vec{j}\in \mathcal{T}^1}p_{\vec{j}}\ket{\psi_{\vec{j}}}\bra{\psi_{\vec{j}}},
\end{equation}
such that $\exists N_1, F(\rho',\rho^{\otimes n})\ge 1-\varepsilon_1,\forall \varepsilon_1, n\ge N_1$.

Next, we consider how to approximate $\ket{\psi_{\vec{j}}}\bra{\psi_{\vec{j}}}$. Note that, $\ket{\psi_{\vec{j}}}\bra{\psi_{\vec{j}}}$ consists of $N_j(\vec{j})\in[n(p_j-\delta_1), n(p_j+\delta_1)]$ copies of $\ket{\psi_j}$. Then, to approximate $\ket{\psi_{\vec{j}}}\bra{\psi_{\vec{j}}}$, we can equivalently to consider how to approximate $\ket{\psi_j}^{\otimes N_j(\vec{j})}$. For a general state $\ket{\psi} = \sum_i \lambda_i \ket{i}$ and $\ket{\psi}^{\otimes n} = \sum_{\vec{i}=i_1i_2\cdots i_n}\lambda_{\vec{i}}\ket{{\vec{i}}}$, we can similarly consider the typical sequence,
\begin{equation}
	\mathcal{T}^2_{\ket{\psi}}=\{\vec{i}:|N_i(\vec{i})/n-|\lambda_i|^2|\le \delta_2, \forall i \}.
\end{equation}
Here $N_i(\vec{i})$ is the number of $i_t$ in $\vec{i}$ such that $i_t=i, \forall t\in\{1,2,\cdots,n\}$. Then, we can approximate $\ket{\psi}^{\otimes n}$ via,
\begin{equation}
	\ket{\psi'}=\sum_{\vec{i}\in \mathcal{T}^2_{\ket{\psi}}}\lambda_{\vec{i}}\ket{{\vec{i}}},
\end{equation}
such that $\exists N_2, F(\ket{\psi'},\ket{\psi}^{\otimes n})\ge 1-\varepsilon_2,\forall \varepsilon_2, n\ge N_2$.

Note that
\begin{equation}
	\ket{\psi_{\vec{j}}} = \pi\left(\prod_j\ket{\psi_j}^{\otimes N_j(\vec{j})}\right),
\end{equation}
where $\pi$ is a permutation between the $n$ copies of the states. Then, we can approximate $\ket{\psi_{\vec{j}}}$ by
\begin{equation}
	\ket{\psi_{\vec{j}}'} = \pi\left(\prod_j\ket{\psi_j'}\right),
\end{equation}
where $\ket{\psi_j'}$ is constructed from the above approximation procedure from $\ket{\psi_j}^{\otimes N_j(\vec{j})}$. We choose a sufficiently large $n$ such that for any $j$, $N_j(\vec{j})> N_2$, thus $F(\ket{\psi_j'},\ket{\psi_j}^{\otimes N_j(\vec{j}) })\ge 1-\varepsilon_2$, and  $F(\ket{\psi_{\vec{j}}'}, \ket{\psi_{\vec{j}}})\ge (1-\varepsilon_2)^{\Omega}$. Here $\Omega$ equals the number of terms in the decomposition of $\rho=\sum_j p_j\ket{\psi_j}\bra{\psi_j}$.

With the above result, we can further approximate $\rho'$ by
\begin{equation}
	\rho'' = \sum_{\vec{j}\in \mathcal{T}}p_{\vec{j}}\ket{\psi_{\vec{j}}'}\bra{\psi_{\vec{j}}'},
\end{equation}
with $F(\rho'',\rho')\ge (1-\varepsilon_2)^\Omega$ by the joint concavity of the fidelity. Therefore, we have $F(\rho'',\rho)\ge (1-\varepsilon_1)(1-\varepsilon_2)^\Omega$ and we can take $\varepsilon = 1-(1-\varepsilon_1)(1-\varepsilon_2)^\Omega$.

Now, we calculate an upper bound to $C_0(\rho'')$ by considering the given decomposition.  To do so, we only need to count the maximal number of nonzero coefficients in $\ket{\psi_{\vec{j}}'}$ or equivalently in $\prod_j\ket{\psi_j'}$. By the law of large numbers, $|\mathcal{T}^2_{\ket{\psi_j}}|$, also the number of nonzero coefficients in $\ket{\psi_j'}$ is upper bounded by $2^{N_j(\vec{j})(H(|\lambda_i|)|+\delta_2)}$. The total number of nonzero coefficients in $\ket{\psi_{\vec{j}}'}$ is upper bounded by
\begin{equation}
\begin{aligned}
	T_j(\ket{\psi_{\vec{j}}'}) &= 2^{\sum_j N_j(\vec{j})(H(|\lambda_i|)|+\delta_2)}\\
	&\le 2^{\sum_j n(p_j+\delta_1)(H(|\lambda_i|)|+\delta_2)}.\\
\end{aligned}
\end{equation}
As $\rho''$ is only a special state that $F(\rho'',\rho)\ge 1-\varepsilon$, we have
\begin{equation}
	C_0^{\varepsilon}(\rho^{\otimes n})\le \sum_j n(p_j+\delta_1)(S(\rho_j^{\textrm{diag}})|+\delta_2),
\end{equation}
with $\rho_j^{\textrm{diag}}= \sum_i\bra{i}\ket{\psi_j}\bra{\psi}\ket{i}\ket{i}\bra{i}$.
Considering all decompositions of $\rho$, we thus have
\begin{equation}
	C_0^{\varepsilon}(\rho^{\otimes n})\le \min_{p_j,\ket{\psi_j}}\sum_j n(p_j+\delta_1)(S(\rho_j^{\textrm{diag}})|+\delta_2),
\end{equation}
Take the limit of $\varepsilon\to 0^+$ and $n\to\infty$, we have
   \begin{equation}
   \begin{aligned}
     C_{IO}^\infty(\rho)&=   \lim_{\varepsilon\to 0^+}\lim_{n\to\infty}\frac{1}{n}C_0^\varepsilon(\rho^{\otimes n})\leq C_f(\rho).
     \end{aligned}
   \end{equation}
\end{proof}

Besides the above proof, Ref.\cite{Buscemi11} provides another possible proof for this lemma by applying Lemma 7 and Lemma 8 in \cite{Buscemi11}.

\section{Additivity}\label{additi}
In this section, we discuss the additivity of the introduced coherence monotones which characterizes the one-shot coherence cost. We provide detailed proofs for $C_0^{\varepsilon}(\rho)=C_{IO}^{\varepsilon}(\rho)$ as an example and the other ones follow similarly.

First of all, we show that $C_0^\varepsilon(\rho_1\otimes\rho_2)\neq C_0^\varepsilon(\rho_1)+C_0^\varepsilon(\rho_2)$ in general. Consider the example of $\rho_1=\rho^{\otimes n-1}$ and $\rho_2=\rho$. Taking the limit of $n\to\infty$ and $\varepsilon\to 0^+$, we obtain
\begin{equation}\label{addlimit1}
  C_0^\varepsilon(\rho^{\otimes n})-C_0^\varepsilon(\rho^{\otimes n-1})\to C_f(\rho),
\end{equation}
and
\begin{equation}\label{addlimit2}
  C_0^\varepsilon(\rho)\to C_0(\rho).
\end{equation}
Since $C_f(\rho)\leq C_0(\rho)$ and for some $\rho$ the inequality is strict, we conclude that $C_0^\varepsilon(\rho)$ is not additive. However, we can provide an upper bound for $C_0^\varepsilon(\rho_1\otimes\rho_2)$.

Next we prove that if $\varepsilon\geq\varepsilon_1+\varepsilon_2$, we have
\begin{equation}\label{addbound}
  C_0^{\varepsilon}(\rho_1\otimes\rho_2)\leq C_0^{\varepsilon_1}(\rho_1)+C_0^{\varepsilon_2}(\rho_2).
\end{equation}
Since $C_0^{\varepsilon}(\rho)=C_{IO}^{\varepsilon}(\rho)$, we prove this inequality from the operational perspective. Suppose we want to prepare the state $\rho_1\otimes\rho_2$ with error $\varepsilon$ via IO. If we prepare it directly, the minimum resource required is $ C_0^{\varepsilon}(\rho_1\otimes\rho_2)$. We consider another preparation method, which is to prepare $\rho_1$ with error $\varepsilon_1$ and $\rho_2$ with error $\varepsilon_2$, respectively. Denote the two result states as $\rho_1'$ and $\rho_2'$, then we regard $\rho_1'\otimes\rho_2'$ as the final approximation of $\rho_1\otimes\rho_2$. This is legal because
\begin{equation}\label{addlegal}
\begin{aligned}
  F(\rho_1\otimes\rho_2,\rho_1'\otimes\rho_2')&=F(\rho_1,\rho_1')F(\rho_2,\rho_2')\\
  &\geq (1-\varepsilon_1)(1-\varepsilon_2)\\
  &\geq 1-(\varepsilon_1+\varepsilon_2)\\
  &\geq 1-\varepsilon.
\end{aligned}
\end{equation}
The resource required for the new preparation method is $C_0^{\varepsilon_1}(\rho_1)+C_0^{\varepsilon_2}(\rho_2)$. By definition of one-shot coherence cost, we obtain $C_0^{\varepsilon}(\rho_1\otimes\rho_2)\leq C_0^{\varepsilon_1}(\rho_1)+C_0^{\varepsilon_2}(\rho_2)$.

\end{document}